\documentclass[reqno]{amsart}
\usepackage{amssymb}

\numberwithin{equation}{section}
\newtheorem{theorem}{Theorem}[section]
\newtheorem{proposition}[theorem]{Proposition}
\newtheorem{lemma}[theorem]{Lemma}
\newtheorem{corollary}[theorem]{Corollary}
\newtheorem{definition}[theorem]{Definition}
\newtheorem{remark}[theorem]{Remark}
\newtheorem{assumption}[theorem]{Assumption}
\newenvironment{acknowledgement}{\emph{Acknowledgement.}}

\renewcommand\H{\mathcal{H}}
\newcommand\W{\mathcal{W}}
\renewcommand\L{\mathrm{L}}
\newcommand\R{\mathbb R}

\newcommand\C{\mathbb C}
\newcommand\Z{\mathbb Z}
\newcommand\K{\mathcal{K}}
\newcommand\D{\mathcal{D}}
\newcommand\Db{\mathbf{D}}
\newcommand\G{\mathcal{G}}
\newcommand\U{\mathcal{U}}
\newcommand\cT{\mathcal{T}}
\newcommand\Ll{\mathcal{L}}
\newcommand\A{\mathbf{A}}
\newcommand\M{\mathcal{M}}
\newcommand\x{\mathbf{x}}

\newcommand\di{\mathrm d}
\renewcommand\v{\mathbf{v}}

\newcommand{\Cc}{C_{\mathrm{c}}}

\newcommand\El{\mathbf E}
\newcommand\J{\mathbf J}
\newcommand\T{\mathcal{T}}
\newcommand\F{\mathbf{F}}

\newcommand\tr{\mathrm{tr}}

\newcommand\e{\mathrm{e}}
\newcommand{\la}{\langle}
\newcommand{\ra}{\rangle}
\renewcommand\P{\mathbb P}
\newcommand\E{\mathbb E}
\newcommand\Zz{\mathcal Z}

\newcommand\eps{\varepsilon}

\begin{document}

\title[Linear response theory for random Schr\"odinger
operators]{Linear response theory for random Schr\"odinger operators
and noncommutative integration}

\author[Nicolas Dombrowski]{Nicolas Dombrowski}
\address[Dombrowski]{ Universit\'e de Cergy-Pontoise,
CNRS UMR 8088, Laboratoire AGM, D\'epartement de Math\'ematiques,
Site de Saint-Martin,
2 avenue Adolphe Chauvin,
F-95302 Cergy-Pontoise, France}
\email{ndombro@math.u-cergy.fr}

\author[F. Germinet]{Fran\c cois Germinet}
\address[Germinet]{ Universit\'e de Cergy-Pontoise,
CNRS UMR 8088, Laboratoire AGM, D\'epartement de Math\'ematiques,
Site de Saint-Martin,
2 avenue Adolphe Chauvin,
F-95302 Cergy-Pontoise, France}
\email{germinet@math.u-cergy.fr}

\thanks{2000 \emph{Mathematics Subject Classification.}
Primary 82B44; Secondary  47B80, 60H25}


\begin{abstract}
We consider an ergodic Schr\"odinger operator
with magnetic field within the non-interacting particle approximation. Justifying the linear response theory, a rigorous
derivation of a Kubo formula for the electric conductivity tensor within this context can be found in a recent work of Bouclet, Germinet, Klein and Schenker \cite{BGKS}. 
If the Fermi level falls into a
region of localization, the well-known Kubo-St\u{r}eda formula for
the quantum Hall conductivity at zero temperature is recovered.
In this review we go along the lines of \cite{BGKS} but make a more systematic use of  noncommutative ${\rm L}^{p}$-spaces, leading to a somewhat more transparent proof.
\end{abstract}

\maketitle

\tableofcontents


\section{Introduction}

In \cite{BGKS} the authors consider an ergodic Schr\"odinger operator with magnetic
field, and give a controlled derivation of a Kubo formula for the electric
conductivity tensor, validating the linear response theory within the
noninteracting particle approximation. For an adiabatically switched electric
field, they then recover the expected expression for the quantum Hall
conductivity whenever the Fermi energy lies either in a region of localization
of the reference Hamiltonian or in a gap of the spectrum.

The aim of this paper is to provide a pedestrian derivation of \cite{BGKS}'s result and to simplify their ``mathematical apparatus" by resorting more systematically to noncommutative  ${\rm L}^{p}$-spaces. We also state results for more general time dependent electric fields, so that AC-conductivity is covered as well. That von Neumann algebra and noncommutative integration play an important r\^ole in the context of random Schr\"odinger operators involved with the quantum Hall effect goes back to Bellissard, e.g. \cite{Be,BESB,SBB1,SBB2}. 

The electric conductivity tensor is usually expressed in
terms of a ``Kubo formula,"  derived via formal linear response theory. In the context of disordered media, where Anderson localization is expected (or proved), the electric conductivity has driven a lot of interest coming for several perspective. For time reversal systems and at zero temperature, the vanishing of the direct conductivity is a physically meaningful evidence of a localized regime \cite{FS,AG}. Another direction of interest is the connection between direct conductivity and the quantum Hall effect \cite{TKNN,St,Be,Ku,BESB,ASS,AG}. On an other hand the behaviour of the alternative conductivity at small frequencies within the region of localization is dictated by the celebrated Mott formula \cite{MD} (see \cite{KLP,KLM,KM} for recent important developments). Connected to conductivities, current-current correlations functions have recently droven a lot of attention as well (see \cite{BH,CGH} and references therein). 

During the past two decades a few papers
managed to shed some light on these derivations from the mathematical point of
view, e.g., \cite{Pa,Ku,Be,NB,ASS,BESB,SBB1,SBB2,AG,Na,ES,CJM,CNP}. While a great
amount of attention has been brought to the derivation from a Kubo formula of conductivities (in particular of the quantum Hall
conductivity), and to the study of these conductivities,
not much has been done concerning a controlled derivation of the linear
response and the Kubo formula itself; only the recent papers
\cite{SBB2,ES,BGKS,CJM,CNP} deal with this question.

In this note, the accent is put on the derivation of the linear response for which we shall present the main elements of proof, along the lines of \cite{BGKS} but using noncommutative integration. The required material is briefly presented or recalled from \cite{BGKS}. Further details and extended proofs will be found in \cite{Do}. We start by describing the noncommutative $\L^p$-spaces that are relevant in our context, and we state the properties that we shall need (Section~\ref{sectnci}). In Section~\ref{sectoperator} we define magnetic random Schr\"odinger operators and perturbations of these by time dependent electric fields, but in a suitable gauge where the electric field is given by a time-dependent
vector potential. We review from \cite{BGKS} the main tools that enter the derivation of the linear response, in particular the time evolution propagators. In Section~\ref{sectkubo} we compute rigorously the linear response
of the system forced by a time dependent electric field. We do it along the lines of \cite{BGKS} but within the framework of the noncommutative $\L^p$-spaces presented in Section~\ref{sectnci}. The derivation is achieved in
several steps. First we set up the Liouville equation which describes the time
evolution of the density matrix under the action of a time-dependent electric
field (Theorem~\ref{thmrho}). In a standard way, this evolution equation can be written as
an integral equation, the so-called Duhamel formula. Second, we compute the
net current per unit volume induced by the electric field and prove that it is
differentiable with respect to the electric field at zero field. This yields in fair generality the desired Kubo formula for the electric conductivity tensor, for any non zero adiabatic parameter (Theorem~\ref{thmsgmjk} and Corollary~\ref{corsgmjk}). 
The adiabatic limit is then performed in order to compute the direct / ac conductivity at zero temperature (Theorem~\ref{thmHall}, Corollary~\ref{corHall} and Remark~\ref{remAC}). In particular we recover the expected expression for the quantum Hall conductivity, the Kubo-St\u{r}eda formula, as in \cite{Be,BESB}. At positive temperature, we note that, while the existence of the electric conductivity \emph{measure} can easily be derived from that Kubo formula \cite{KM}, proving that  the conductivity itself, i.e. its density, exists and is finite remains out of reach.

\begin{acknowledgement} 
We thank warmly Vladimir Georgescu for enlightening discussions on noncommutative integration, as well as A. Klein for useful comments.
\end{acknowledgement}


\section{Covariant measurable operators and noncommutative ${\rm L}^{p}$-spaces}
\label{sectnci}

In this section we construct the noncommutative ${\rm L}^{p}$-spaces that are relevant for our analysis. We first recall the underlying Von Neumann alegbra of observables and we equip it with the so called ``trace per unit volume". We refer to \cite{Dix,Te} for the material. We shall skip some details and proofs for which we also refer to \cite{Do}.

\subsection{Observables}
Let $\mathcal{H}$ be a separable Hilbert space (in our context $\H={\rm L}^2(\R^d)$). Let $\mathcal{Z}$ be an abelian locally compact group and $ U=\lbrace U_{a}\rbrace_{a\in\mathcal{Z}}$  a  unitary projective representation of $\mathcal{Z}$ on $\mathcal{H}$, i.e. 
\begin{itemize}
\item$\  U_{a} U_{b} = \xi(a,b)U_{a+b}$, where $\xi(a,b)\in\C$, $|\xi(a,b)|=1$;
\item$\ U_{e}=1$;
\end{itemize}
Now we take a set of orthogonal projections on $\mathcal{H}$ , $\chi:=\lbrace \chi_{a} \rbrace _{a \in\mathcal{Z}}$, $\mathcal{Z}\rightarrow\mathcal{B}(\mathcal{H})$.
Such that if a $\neq b \Rightarrow \chi_{a} \chi_{b} = 0 $ and 
$\sum _{a \in\mathcal{Z}} $  
$\chi_a $  =1. Furthermore one requires a covariance relation or a stability relation of $\chi $
under $U$  i.e 
$\ U_{a}\chi_{b}U^{*}_{a}=\chi_{a+b}$.

Next to this Hilbertian structure (representing the ``physical" space), we consider a probability space $(\Omega,\mathcal{F},\P)$ (representing the presence of the disorder) that is ergodic under the action of a group $\tau=\lbrace\tau_{a} \rbrace_{a \in\mathcal{Z}}$, 
that is,
\begin{itemize}
\item $\forall a\in\Zz, \ \tau_{a}:\Omega\rightarrow\Omega$ is a measure preserving isomorphism;
\item $\forall a,b\in\Zz, \ \tau_{a} \circ\tau_{b}= \tau_{a +b} $;  
\item $\tau_{e}=1 $ where $e$ is the neutral element of $\mathcal{Z}$  and thus $\tau_{a}^{-1}= \tau_{-a}, \ \forall a\in\Zz $;
\item If $\mathcal{A}\in \mathcal{F}$ is invariant under $\tau$, then $\P(A)=0$ or $1$.
\end{itemize}

With these two structures at hand we define the Hilbert space
\begin{equation}
\tilde{\mathcal{H}}= \int^{\oplus}_{\Omega}\mathcal{H}\, \mathrm{d}\P(\omega):=\L^{2}(\Omega,\P,\mathcal{H})\simeq \mathcal{H}\otimes\L^{2}(\Omega,\P),
\end{equation}
equipped with the inner product
\begin{equation}
\langle\varphi,\psi\rangle_{\tilde{\H}} = \int_{\Omega}\langle\varphi(\omega),\psi(\omega)\rangle_{\H} \, \mathrm{d}\P(\omega), \
\forall\varphi,\psi\in\tilde{\mathcal{H}^{2}}.
\end{equation}

We are interested in bounded operators on $\tilde{\mathcal{H}}$ that are decomposable elements $A=(A_\omega)_{\omega\in\Omega}$, in the sense that they commute with the diagonal ones. Measurable operators are defined as decomposable operators such that for all measurable vector's field $\lbrace\varphi(\omega),\omega\in\Omega\rbrace$, 
the set $\lbrace A(\omega)\varphi(\omega),\omega\in\Omega\rbrace$
is measurable too.
Measurable decomposable operators are called essentially bounded if $\omega\rightarrow\Vert A_{\omega}\Vert_{\mathcal{L}({\mathcal{H}})}$
is a element of $ L^{\infty}(\Omega,\P)$. We set, for such $A$'s,
\begin{equation}
\Vert{A}\Vert_{\mathcal{L}(\tilde{\mathcal{H}})}= \Vert A \Vert_{\infty}= \mathrm{ess-sup}_{\Omega}\Vert A(\omega)\Vert ,
\end{equation}
and define the following von Neumann algebra
\begin{equation}
\mathcal{K} = \ L^{\infty}(\Omega,\P,\mathcal{L}(\mathcal{H})) = \lbrace A:\Omega\rightarrow\mathcal{L}(\mathcal{H}), measurable \ \| A \|_{\infty} <
\infty\rbrace.
\end{equation}
There exists an isometric isomorphism betwen
$\mathcal{K}$ and decomposable operators on 
$\mathcal{L}(\tilde{\mathcal{H}})$. 

We shall work with observables of $\mathcal{K}$ that satisfy the so-called covariant property.

\begin{definition}
$A\in\mathcal{K}$ is covariant iff
\begin{equation}\label{cov}
U_{a}A(\omega) U^{\star}_{a}=A(\tau_a\omega)\ ,\forall a\in\mathcal{Z},\,\forall\omega\in\Omega.
\end{equation}  
We set
\begin{equation}
\mathcal{K}_{\infty}=\lbrace A\in\mathcal{K},  A \mbox{ is covariant}\rbrace.
\end{equation}
\end{definition}

If $\tilde{U_{a}}:= U_{a}\otimes\tau(-a)$,with the slight notation abuse where we note $\tau$ for the action induct by $\tau$ on $\L^{2}$ and $\tilde{U}=(\tilde{U_{a}})_{a\in\Zz}$, we note that
\begin{eqnarray}
\mathcal{K}_{\infty}& = &\lbrace A\in\mathcal{K}, \, \forall a\in\mathcal{Z}, [A,\tilde{U}_{a}]=0\rbrace\\
&=&\mathcal{K}\cap (\tilde{{U}})',
\end{eqnarray}
so that $\mathcal{K}_{\infty}$ is again a von Neumann algebra.

\subsection{Noncommutative integration}

The von Neumann algebra $\mathcal{K}_{\infty}$ is equipped with the faithfull, normal and semi-finite trace
\begin{equation}
\mathcal{T}(A) :=\E\lbrace \tr(\chi_{e}A(\omega)\chi_{e})\rbrace,
\end{equation}
where ``$\tr$" denotes the usual trace on the Hilbert space $\H$.In the usual context of the Anderson model this is nothing but the  trace per unit volume,  by the Birkhoff Ergodic Theorem, whenever 
$\mathcal{T}(|A|)<\infty$, one has
\begin{equation}
\mathcal{T}(A)=\lim_{|\Lambda_{L}|\to\infty}\dfrac{1}{\vert\Lambda_{L}\vert}\tr(\chi_{\Lambda_{L}}A_{\omega}\chi_{\Lambda_{L}}),
\end{equation}
where $\Lambda_{L}\subset\mathcal{Z}$ and $\chi_{\Lambda_{L}}=\sum_{a\in\Lambda_{L}} \chi_a$. There is a natural topology associated to von Neumann algebras equipped with such a trace. It is defined by the following basis of neighborhoods:
\begin{equation}
N(\epsilon ,\delta)=\lbrace A\in\mathcal{K}_{\infty},\,  \exists P\in\mathcal{K}^{proj}_{\infty} ,\Vert AP\Vert_{\infty}<\eps\ ,\mathcal{T}(P^{\perp})<\delta\rbrace,
\end{equation}
where $\mathcal{K}^{proj}_{\infty}$ denotes the set of projectors of 
$\mathcal{K}_{\infty}$. It is a well known fact that 
\begin{equation}
A\in N(\eps,\delta)\Longleftrightarrow\mathcal{T}(\chi_{]\eps,\infty[}(\vert A\vert))\leq\delta.
\end{equation}
As pointed out to us by V. Georgescu, this topology can also be generated by the following Frechet-norm on $\mathcal{K}_{\infty}$ \cite{Geo}:
\begin{equation}
\| A\|_{\mathcal{T}}= \inf_{P\in\mathcal{K}^{proj}_{\infty}}\max\lbrace\Vert AP\Vert_{\infty} , \mathcal{T}(P^{\perp})\rbrace.
\end{equation}
Let us denote by $\mathcal{M}(\mathcal{K}_{\infty})$ the set of all $\mathcal{T}$-measurable operators, namely
the completion of $\mathcal{K}_{\infty}$ with respect to this topology.
It is a well established fact from noncommutative integration that
\begin{theorem}
$\mathcal{M}(\mathcal{K}_{\infty})$ is a Hausdorff topological $\ast$-algebra.
, in the sens that all the algebraic operations are continuous for the $\tau$-measure topology.
\end{theorem}
\begin{definition}
A linear subspace $\mathcal{E}\subseteq\mathcal{H}$ is called $\mathcal{T}$-dense if, 
$\forall\delta\in\R^{+}$, there exists a projection $P\in\mathcal{K}_{\infty}$ 
such that $\ P\mathcal{H}\subseteq\mathcal{E}$ and $\mathcal{T}(P^{\perp})\leq\delta.$
\end{definition}

It turns out that any element $A$ of $\mathcal{M}(\mathcal{K}_{\infty})$ can be represented as an (possibly unbounded) operator, that we shall still denote by $A$, acting on $\tilde{\mathcal{H}}$ with a domain $D(A)=\lbrace\varphi\in\tilde{\mathcal{H}}, \ A\varphi\in\tilde{\mathcal{H}}\rbrace$ that is $\mathcal{T}$-densily defined.
Then, adjoints, sums and products of elements of $\mathcal{M}(\mathcal{K}_{\infty})$ are defined as usual adjoints, sums and products of unbounded operators.

For any $0<p<\infty$, we set
\begin{equation}
\mathrm{L}^{p}(\mathcal{K}_{\infty}) := \overline{\lbrace x\in K_{\infty}, \ \mathcal{T}( |x|^p )<\infty \rbrace}^{\|\cdot\|_p} =
\lbrace x\in\mathcal{M}(\mathcal{K}_{\infty}) , \  \mathcal{T}(\vert x\vert^p)<\infty \rbrace,
\end{equation}
where the second equality is actually a theorem.
For $p\ge 1$, the spaces $\mathrm{L}^{p}(\mathcal{K}_{\infty})$ are  Banach spaces in which  $\mathrm{L}^{p,o}(\mathcal{K}_{\infty}) :=\mathrm{L}^{p}(\mathcal{K}_{\infty})\cap \mathcal{K}_{\infty}$ are dense by definition. For $p=\infty$, in analogy with the commutative case, we set $\mathrm{L}^{\infty}(\mathcal{K}_{\infty})=\mathcal{K}_{\infty}$.

Noncommutative H\"older inequalities hold: for any $0<,p,q,r\le \infty$ so that $p^{-1} +q^{-1}=r^{-1}$, if $A\in\mathrm{L}^{p}(\mathcal{K}_{\infty})$ and $B\in\mathrm{L}^{q}(\mathcal{K}_{\infty})$, then the product $AB\in\mathcal{M}(\mathcal{K}_{\infty})$ belongs to $\mathrm{L}^{r}(\mathcal{K}_{\infty})$ with
\begin{equation}\label{Holder}
\Vert AB\Vert_{r}\leq\Vert A\Vert_{p}\Vert B\Vert_{q}.
\end{equation}
In particular, for all $A\in\mathrm{L}^{\infty}(\mathcal{K}_{\infty})$ and $B\in\mathrm{L}^{p}(\mathcal{K}_{\infty})$,
\begin{equation}
\Vert AB\Vert_{p}\leq\Vert A\Vert_{\infty}\Vert B\Vert_{p} 
\mbox{ and } 
\Vert BA\Vert_{p}\leq\Vert A\Vert_{\infty}\Vert B\Vert_{p}\, ,
\end{equation}
so that $\mathrm{L}^{p}(\mathcal{K}_{\infty})$-spaces are $\mathcal{K}_{\infty}$ two-sided submodules of $\mathcal{M}(\mathcal{K}_{\infty})$.
As another consequence, bilinear forms 
$
\mathrm{L}^{p,o}(\K_\infty)\times \mathrm{L}^{q,o}(\K_\infty)\ni(A,B)\mapsto \mathcal{T}(AB)\in\C
$ 
continuously extends to  bilinear maps defined on $\mathrm{L}^{p}(\K_\infty)\times \mathrm{L}^{q}(\K_\infty)$.

\begin{lemma}\label{lemduality}
Let $A\in\mathrm{L}^{p}(\K_\infty)$,  $p\in [1,\infty[$ be given, and suppose $\T(AB)=0$
 for all $B\in\mathrm{L}^q(\K_\infty)$, $p^{-1}+q^{-1}=1$. Then $A=0$.
\end{lemma}

The case $p=2$ is of particular interest since $\mathrm{L}^{2}(\mathcal{K}_{\infty})$ equipped with the sesquilinear form $\langle A, B \rangle_{\mathrm{L}^{2}} = \cT(A^\ast B)$ is a Hilbert space. The corresponding norm reads
\begin{equation}
\Vert A \Vert_{2}^{2} =\int_{\Omega} \tr(\chi_{e} A^{\ast}_\omega A_\omega \chi_{e})\mathrm{d}\P(\omega)
 = \int_{\Omega}\Vert A_\omega\chi_{e}\Vert_{2}^{2}\mathrm{d}\P(\omega).
\end{equation}
(Where $\|\cdot\|_2$ denotes the Hilbert-Schmidt norm.) From the case $p=2$, we can derive the centraliy of the trace. Indeed, by covariance and using the centrality of the usual trace, it follows that $\T(AB)=\T(BA)$ whenever $A,B\in\K_\infty$. By density we get the following lemma.

\begin{lemma}\label{lemcyclicity}
Let $A\in\mathrm{L}^{p}(\K_\infty)$ and $B\in\mathrm{L}^q(\K_\infty)$, $p^{-1}+q^{-1}=1$ be given. Then $\T(AB)=\T(BA)$.
\end{lemma}

We shall also make use of the following observation.

\begin{lemma}\label{lemlimstrong}
Let $A\in\mathrm{L}^{p}(\K_\infty)$ and $(B_n)$ a sequence of elements of $\K_\infty$ that  converges strongly to $B\in\K_\infty$. Then $AB_n$ converges to $AB$ in $\mathrm{L}^{p}(\K_\infty)$.
\end{lemma}


\subsection{Commutators  of measurable covariant operators}
Let $H$ be a decomposable (unbounded) operator affiliated to $\K_\infty$ with domain $\D$, and $A\in\M(\K_\infty)$. In particular $H$ need not be $\T$-measurable, i.e. in $\M(\K_\infty)$. If there exists a $\T$-dense domain $\D'$  such that $\ A\D'\subset \D$, then $HA$ is well defined, and if in addition the product is $\T$-measurable then we write $HA\in \M(\K_\infty)$. Similarly, if $\D$ is $\T$-dense and the range of $H\D\subset\D(A)$, then $AH$ is well defined, and if in addition the product is $\T$-measurable then we write $AH\in\M(\K_\infty)$.

\begin{remark}
We define the following (generalized)  commutators:
\begin{description}
\item[(i)]
If $ A\in\M\ (\K_{\infty})$ and $B\in\K_\infty$, 
\begin{equation}
[A,B]  :=  AB - BA \in\M(\K_\infty) , \quad [B,A]:=-[A,B].
\end{equation}
\item[(ii)]
If $A \in\mathrm{L}^p(\K_\infty)$, $B\in\mathrm{L}^q(\K_\infty) $, $p,q\ge 1$ such that $p^{-1} + q^{-1}=1$ , then
\begin{equation}
[A,B]  :=
 A B - B A \in\mathrm{L}^1(\K_\infty).
\end{equation}
\end {description}
\end{remark}

\begin{definition}\label{defCommutator}
Let $H \eta\K_\infty$(i.e $H$ affiliated to $\K_\infty$). If $A\in\M(\K_\infty)$ is such that $HA$ and $AH$ are  in $\M(\K_\infty)$, then
\begin{equation}\label{defcomHA}
[H,A] := H A - AH \in\M(\K_\infty) \, .
\end{equation}

\end{definition}

We shall need the following observations.

\begin{lemma}\label{formulaCommutator}
1)  For any $A\in \mathrm{L}^p(\K_\infty), B\in \mathrm{L}^q(\K_\infty)$, $p,q\ge 1$,  $p^{-1}+q^{-1}=1$, and $C_\omega \in \K_\infty$, we have
\begin{equation}
\T \left\{  [C, A] B\right\}=
\T \left\{ C [A, B]\right\} .
\end{equation}
2) For any $A, B\in
\K_\infty$ and $ C\in\mathrm{L}^1(\K_\infty)$, we have
\begin{equation}
\T \left\{ A  [B, C] \right\} =
\T \left\{ [A, B]   C \right\} .
\end{equation}
3) Let $p,q\ge 1$ be such that $p^{-1}+q^{-1}=1$. For any $A\in \mathrm{L}^p(\K_\infty)$, resp. $B\in \mathrm{L}^q(\K_\infty)$, such that $[H,A]\in \mathrm{L}^p(\K_\infty)$, resp. $[H,B]\in \mathrm{L}^q(\K_\infty)$, we have
\begin{equation}\label{LLcycle}
\T \left\{[H,A]  {B}\right\} =
 - \T\left\{A  [H,{B}]\right\}.
\end{equation}

\end{lemma}

\subsection{Differentiation}

A $\ast$-derivation $\partial$ is a $\ast$-map defined  on a dense sub-algreba of $\K_\infty$ and such that:
\begin{itemize}
\item $\partial(AB)=\partial(A)B +A\partial(B)$\
\item $\partial(A+\lambda B)=\partial(A) +\lambda\partial(B)$\
\item $\partial(A^{\star})=\partial(A)^{\star}$
\item $\ [\alpha_{a},\partial]=0$ in the sense that $\alpha_{a}\circ\partial(A)=\partial\circ\alpha_{a}(A)\
\forall a\in\mathcal{Z}\ , \forall A\in\mathcal{K}_{\infty}$.
\end{itemize}

If $\partial_{1},...,\partial_{d}$ are $\ast$-derivations we define a non-commutative gradient by $\nabla:=(\partial_{1},...,\partial_{d})$, densily defined on $\mathcal{K}_{\infty}$. We define  a  non-commutative Sobolev space 
\begin{equation}
\W^{1,p}(\K_\infty):=\lbrace A\in \mathrm{L}^p(\K_\infty) , \,\nabla A\in\mathrm{L}^p(\K_\infty)\rbrace.
\end{equation}
 and a second space for $ H\eta\K_{\infty}$,
\begin{equation}
\D^{(0)}_{p}(H)=
\left\{A\in \mathrm{L}^p(\K_\infty), \;\;
H A, AH \in \mathrm{L}^p(\K_\infty)\right \}.
\end{equation}

\section{The setting: Schr\"odinger operators and dynamics}
\label{sectoperator} \setcounter{equation}{0}

In this section we describe our background operators and recall from \cite{BGKS} the main properties we shall need in order to establish the Kubo formula, but within the framework of noncommutative integration when relevant (i.e. in Subsection~\ref{subsectMSOrandom}).

\subsection{Magnetic Schr\"odinger operators and time-dependent operators}
\label{subsectMSO}
~\\
Throughout this paper we shall consider Schr\"odinger operators of general form
\begin{eqnarray}  \label{MSO}
H= H(\A,V) = \left(-i\nabla - \A\right)^2 + V  \; \; \; \mathrm{on} \; \; \;
\mathrm{L}^2(\mathbb{R}^d),
\end{eqnarray}
where the magnetic potential $\A$ and the electric potential $V$
 satisfy the Leinfelder-Simader conditions:
\begin{itemize}
\item   $\A(x) \in \mathrm{L}^4_{\mathrm{loc}}(\R^d; \R^d)$  with
$\nabla \cdot \A(x) \in \mathrm{L}^2_{\mathrm{loc}}(\R^d)$.

\item  $V(x)= V_+(x) - V_-(x)$ with
 $V_\pm(x) \in \mathrm{L}^2_{\mathrm{loc}}(\R^d)$, $V_\pm(x) \ge 0$,
 and
$ V_-(x)$  relatively bounded with respect to
$\Delta$ with relative bound $<1$, i.e., there are $0 \le\alpha  < 1$ and $\beta \ge 0$
such that
\begin{equation}\label{relbound}
 \|V_-\psi\| \leq
\alpha \| \Delta \psi \| + \beta \|\psi \| \quad \mbox{for all $\psi \in\D(\Delta)$}.
\end{equation}
\end{itemize}
  Leinfelder and Simader have shown that $ H(\A,V)$
 is essentially self-adjoint on $\Cc^\infty(\R^d)$
 \cite[Theorem 3]{LS}. It has been checked in \cite{BGKS} that under these hypotheses $H({\A},V)$ is bounded from below:
\begin{equation} \label{lowerbound}
H({\A},V) \ge
 -\, \frac \beta {(1 - {\alpha})} =: -\gamma +1,  \mbox{ so that } H + \gamma \ge 1.
\end{equation}

We denote by $x_j$ the multiplication operator in $\mathrm{L}^2(\R^d)$ by the $j^{\rm th}$ coordinate $x_j$, and $\x:=(x_1,\cdots x_d)$. 
We want to implement the adiabatic switching of a time dependent spatially uniform electric field $ \El_\eta(t)\cdot\x=\e^{\eta t} \El(t)\cdot \x$ between time $t=-\infty$ and time $t=t_0$. Here $\eta>0$ is the adiabatic parameter and we assume that
\begin{equation}
\int_{-\infty}^{t_0} \e^{\eta t} |\El(t)| \mathrm{d}t < \infty.
\end{equation}
To do so we consider the time-dependent magnetic potential $\A(t)=\A + \F_\eta(t)$, with $\F_\eta(t) = \int_{-\infty}^t \El_\eta(s)\di s$. In other terms, the dynamics is generated by the time-dependent magnetic operator
\begin{eqnarray}\label{eq:Htilde}
    H(t) = ({-i\nabla} - \A - \F_\eta(t))^2  + V(x)=  H(\A(t),V)  \, ,
\end{eqnarray}
which is essentially self-adjoint on
 $\Cc^\infty(\R^d)$ with domain $\D :=  \D(H)=  \D(H(t))$ for all $t \in \R$. One has (see \cite[Proposition~2.5]{BGKS})
\begin{equation}
H(t) = H - 2 \F_\eta(t)\cdot{\Db}(\A) + \F_\eta(t)^2 \;\;\mbox{on  $\D(H)$} ,\label{HtH}
\end{equation}
where $\Db=\Db(\A)$ is the closure of $ (-i\nabla - \A)$ as an
operator from $\mathrm{L}^2(\R^d)$ to $\mathrm{L}^2(\R^d; \C^d)$
with domain $\Cc^\infty(\R^d)$. Each of its components $\Db_j=
\Db_j(\A) =  (-i \frac \partial {\partial x_j} - \A_j)$,
$j=1,\ldots,d$, is essentially self-adjoint on $\Cc^\infty(\R^d)$.

To see that such a family of operators generates the dynamics a quantum particle in the presence of the time dependent spatially uniform electric field $ \El_\eta(t)\cdot\x$, consider the gauge transformation
\begin{equation}\label{gaugedef}
[G(t)\psi](x) := \e^{i \F_\eta(t) \cdot x} \psi(x) \; ,
\end{equation}
so that
\begin{equation}\label{gaugeC}
    H(t) \ = \ G(t) \left [ (-i\nabla - \A)^2 + V \right ] G(t)^*\; .
\end{equation}
Then if $\psi(t)$ obeys Schr\"odinger equation
\begin{equation}\label{schreq1}
i \partial_t \psi(t) = H(t)
\psi(t),
\end{equation}
one has, \emph{formally},
\begin{equation}
 i \partial_t G(t)^\ast \psi(t) = \left [ (-i \nabla - \A)^2 + V +  \El_\eta(t)\cdot x \right ] G(t)^\ast \psi(t) .
\end{equation}

To summarize the action of the gauge transformation we recall the
\begin{lemma}\cite[Lemma~2.6]{BGKS}
\label{lemDA} Let $G(t)$ be as  in \eqref{gaugedef}.  Then
\begin{eqnarray}
G(t)\D&=&\D \, ,\\
\label{gaugeC1}
    H(t) &=&  G(t) H G(t)^*\, ,\\
\Db(\A + \F_\eta(t)) &=&  \Db(\A) - \F_\eta(t)= G(t) \Db(\A) G(t)^* \label{defDAt}.
\end{eqnarray}
Moreover, $i[H(t),x_j]=2 \Db(\A + \F_\eta(t))$ as quadratic forms on
$\D\cap\D(x_j)$, $j=1,2, \dots,d$.
\end{lemma}

The key observation is that the general theory
of propagators with a time dependent generator \cite[Theorem
XIV.3.1]{Y} applies to $H(t)$. It thus yields the existence of a two parameters family $U(t,s)$ of unitary operators,  jointly
strongly continuous in $t$ and $s$, that solves the Schr\"odinger equation.
\begin{eqnarray}  \label{U1}
U(t,r) U(r,s)&= &U(t,s) \\
U(t,t)&= & I\\
\label{Dinv}
U(t,s) \D &=& \D \, ,\\
i \partial_t U(t,s) \psi &=&  H(t) U(t,s) \psi \;\; \mbox{for all $\psi \in \D$}\, ,  \label{leftdiff}\\
\label{reverseddif}
 i \partial_s U(t,s) \psi &=& -\, U(t,s) H(s) \psi \;\; \mbox{for all $\psi \in \D$}\, .
\end{eqnarray}
We refer to \cite[Theorem~2.7]{BGKS} for other relevant properties.

To compute the linear response, we shall make use of the following
``Duhamel formula". Let $ U^{(0)}(t) = \mathrm{e}^{-itH} $. For all  $ \psi \in \D $
and  $ t,s \in \R $ we have \cite[Lemma~2.8]{BGKS}
\begin{equation}\label{DuhamelU}
U(t,s) \psi = U^{(0)}(t-s)\psi  + i \int_{s}^t U^{(0)}(t-r) (2  \F_\eta(r)\cdot
\Db(\A)-\F_\eta(r)^2)U(r,s) \psi \ \mathrm{d}r  \, .
\end{equation}
Moreover,
\begin{equation}\label{DuhamelU2}
 \lim_{|\El| \rightarrow 0} U(t,s) = U^{(0)}(t-s) \;\;\mbox{strongly}\, .
\end{equation}


\subsection{Adding the randomness}
\label{subsectMSOrandom}

Let $(\Omega, \P)$ be a probability space equipped with an ergodic
group $\{\tau_a; \ a \in \Z^d\}$ of measure preserving
transformations.  We study operator--valued maps  $A \colon \Omega \ni\omega\mapsto A_\omega$.

Throughout the rest of this paper we shall use the material of Section~\ref{sectnci} with $\H=\mathrm{L}^2(\R^d)$ and $\mathcal{Z}=\Z^d$. The projective representation of $\Z^d$ on $\H$ is given by magnetic translations $(U(a)\psi)(x) = \e^{ia\cdot Sx}\psi(x-a)$, $S$ being a given $d \times d$ real matrix. The projection $\chi_a$ is the characteristic function of the unit cube of $\R^d$ centered at $a\in\Z^d$.

In our context natural $\ast$-derivations arise: 
\begin{equation}
\partial_j A := i[x_j, A], \; j=1,\cdots,d, \quad \nabla A=i[\x,A].
\end{equation}

We shall now recall the material from \cite[Section~4.3]{BGKS}. Proofs of assertions are extensions of the arguments of \cite{BGKS} to the setting of $\mathrm{L}^p(\K_\infty)$-spaces. We refer to \cite{Do} for details.

We state the technical assumptions on our  Hamiltonian of reference $H_\omega$.
\begin{assumption}\label{RandH}
  The ergodic Hamiltonian $\omega \mapsto H_\omega$ is a measurable map
 from the
  probability space $(\Omega, \P)$ to  self-adjoint operators on $\H$ such
  that
  \begin{equation}
    H_\omega= H(\A_\omega,V_\omega) = \left(-i\nabla - \A_\omega\right)^2 +
    V_\omega \; ,
  \end{equation}
  almost surely, where $\A_\omega$ ($V_\omega$) are vector (scalar) potential
  valued random variables which satisfy the Leinfelder-Simader conditions
(see Subsection~\ref{subsectMSO})
  almost surely. It is furthermore assumed that $H_\omega$ is covariant as in \eqref{cov}. We denote by $H$ the operator $(H_\omega)_{\omega\in\Omega}$ acting on $\tilde{\H}$.
\end{assumption}

As a consequence $\|f(H_\omega)\|\le \|f\|_\infty$ and  $f(H)\in \K_\infty$ for every bounded Borel function $f$ on
the real line. In particular $H$ is affiliated to $\K_\infty$.
For $\P$-a.e. $\omega$, let $U_\omega(t,s)$ be the  unitary propagator associated to $H_\omega$ and described in Subsection~\ref{subsectMSO}. Note that $(U_\omega(t,s))_{\omega\in\Omega} \in \K_\infty$ (measurability in $\omega$ follows by construction of $U_\omega(t,s)$, see \cite{BGKS}).
For $A \in \M(\K_{\infty})$ decomposable, let
\begin{eqnarray}  \label{UtsUst}
{\U}(t,s) (A) :=\int_{\Omega}^{\oplus} U_\omega(t,s)A_\omega U_\omega(s,t) \, \mathrm{d}\P(\omega) .
\end{eqnarray}
Then ${\U}(t,s)$ extends a linear operator on $\M(\K_{\infty})$, leaving invariant $\M(\K_{\infty})$ and $\mathrm{L}^p(\K_\infty)$, $p\in]0,\infty]$, with
\begin{eqnarray}
{\U}(t,r){\U}(r,s) &=& {\U}(t,s) \, ,\\
{\U}(t,t) &=&  I \,,\\
\left\{{\U}(t,s) (A)\right\}^\ast&=&
 {\U}(t,s) (A^\ast)\, \label{UAUdag} .
\end{eqnarray}
Moreover, ${\U}(t,s)$ is to a unitary  on $\mathrm{L}^2(\K_\infty)$ and an isometry in
 $\mathrm{L}^p(\K_\infty)$, $p\in [1,\infty]$. In addition,
 ${\U}(t,s)$ is jointly strongly continuous
in $t$ and $s$ on $\mathrm{L}^p(\K_\infty)$, $p\in 1,\infty]$.

Pick $p>0$. Let
 $A\in \mathrm{L}^p(\K_\infty)$ be such that
  $H(r_0)A$ and
$AH(r_0)$ are  in $\mathrm{L}^p(\K_\infty)$ for some
 $r_0 \in [-\infty, \infty)$.
Then both maps $r\mapsto {\U}(t,r)(A)\in \mathrm{L}^p(\K_\infty)$ and $t\mapsto {\U}(t,r)(A)\in \mathrm{L}^p(\K_\infty)$ are differentiable in
 $\mathrm{L}^p(\K_\infty)$, with (recalling Definition~\ref{defCommutator})
\begin{eqnarray}
i\partial_r\,  {\U}(t,r)(A) & = &
-{\U}(t,r)([H(r),A]) . \label{eqderivtau}
\\
 i\partial_t\,  {\U}(t,r)(A)
& = &
[H(t),{\U}(t,r)(A)]  . \label{eqderivU3}
\end{eqnarray}
Moreover, for $t_0<\infty$ given, there exists $C=C(t_0)<\infty$ such that for all $t,r\le t_0$,
\begin{gather}\label{HUAW}
\|\left(H(t)+ \gamma\right){\U}(t,r)(A)\|_p \le
C
\| (H(r)+\gamma) A\|_p \, ,\\
\label{HUAWbound}
\left\|[H(t),{\U}(t,r)(A)]\right\|_p \le
  C \left( \| (H(r)+\gamma) A\|_p
+
 \| A (H(r)+\gamma)\|_p
\right).
\end{gather}
We  note that in order to apply the above formula and in particular \eqref{eqderivtau} and \eqref{eqderivU3}, it is actually enough to check that $(H(r_0)+\gamma)A$ and $A(H(r_0)+\gamma)$ are  in $\mathrm{L}^p(\K_\infty)$.

Whenever we  want to keep track of the dependence of  $U_\omega(t,s)$
on the electric field $\El=\El_\eta(t)$,  we shall write $U_\omega(\El, t,s)$.  
When  $\El=0$, note that
\begin{equation} \label{U0}
U_\omega(\El= 0, t,s)  = U^{(0)}_\omega(t-s) := \e^{-i(t-s) H_\omega}  .
\end{equation}
For $A \in \M(\K_{\infty})$ decomposable, we let
\begin{equation}  \label{U00}
\U^{(0)}(r)(A):=\int_{\Omega}^{\oplus} U^{(0)}_\omega(r)  A_\omega U^{(0)}_\omega(-r) \, \mathrm{d}\P (\omega) .
\end{equation}
We still denote by $\U^{(0)}(r)(A)$ its extension to $\M(\K_{\infty})$.

\begin{proposition}\label{Liouvillian}
Let $p\ge 1$ be given. ${\U}^{(0)}(t)$ is a one-parameter group of operators
on $\M(\K_{\infty})$, leaving   $\mathrm{L}^p(\K_\infty)$   invariant.
${\U}^{(0)}(t)$ is an isometry on $\mathrm{L}^p(\K_\infty)$, and unitary if $p=2$.
It is strongly continuous on  $\mathrm{L}^p(\K_\infty)$. We further denote by $\Ll_{p}$ the  infinitesimal generator of ${\U}^{(0)}(t)$ in $\mathrm{L}^p(\K_\infty)$:
\begin{equation}
\U^{(0)}(t) = \e^{-it \Ll_{p}} \quad \mbox{for all $t \in \R$}\, .
\end{equation}
The operator $\mathcal{L}_p$ is usually called the \emph{Liouvillian}.
Let
\begin{equation}\label{domainL0}
\D^{(0)}_{p}=
\left\{A\in \mathrm{L}^p(\K_\infty), \;\;
H A, AH \in \mathrm{L}^p(\K_\infty)\right \}.
\end{equation}
Then $\D^{(0)}_{p}$ is an operator core for $\Ll_{p}$ (note that
$\Ll_{2}$ is essentially self-adjoint on
$\D^{(0)}_{2}$), and
\begin{equation}
\Ll_{p} (A) = [ H, A]  \quad
\mbox{for all $A \in \D^{(0)}_{p}$} .
\end{equation}
Moreover, for every $B_\omega \in \K_\infty $ there exists a sequence
$B_{n,\omega} \in\D^{(0)}_{\infty} $ such that $B_{n,\omega} \to B_\omega$
as a  bounded and $\P$-a.e.-strong limit.
\end{proposition}

 We finish this list of properties with the following lemma about the Gauge transformations in spaces of measurable operators.
The map
\begin{eqnarray}\label{Gtextends}
\G(t) (A) =  G(t) A G(t)^* \; ,
\end{eqnarray}
with $G(t)=\e^{i \int_{-\infty}^t \mathbf{E}_\eta(s)\cdot \x}$ as in \eqref{gaugedef}, is an
isometry on $\mathrm{L}^p(\K_\infty)$, for $p\in]0,\infty]$.

\begin{lemma} \label{lemGt}
For any $p\in]0,\infty]$, the map $\G(t)$ is strongly continuous on  $\mathrm{L}^p(\K_\infty)$, and
\begin{equation}\label{liminfty}
\lim_{t \to -\infty} \G(t) = I \;\;\mbox{strongly}
\end{equation}
on  $\mathrm{L}^p(\K_\infty)$. Moreover, if $A\in\W^{1,p}(\K_\infty)$, then $\G(t)
(A)$ is continuously differentiable in $\mathrm{L}^p(\K_\infty)$ with
\begin{equation}\label{eqderivK}
\partial_t \G(t)(A) ={\El}_\eta(t)\cdot\nabla(\G(t)(A)).
\end{equation}
\end{lemma}


\section{Linear response theory and Kubo formula}
\label{sectkubo}

\subsection{Adiabatic switching of the electric field}
\setcounter{equation}{0}

We now fix an initial equilibrium state of the system, i.e., we specify a
density matrix ${\zeta}_\omega$ which is in equilibrium, so $[H_\omega,
{\zeta}_\omega ] =0$. For physical applications, we would generally take
${\zeta}_\omega=f(H_\omega)$ with $f$
the Fermi-Dirac distribution at inverse temperature $\beta \in (0,\infty]$ and
\emph{Fermi energy} $E_F \in \R$,  i.e.,
$f(E) = \frac{1}{1+\e^{\beta (E - E_F)}}$ if $\beta < \infty$ and
$f(E)= \chi_{(-\infty,E_F]}(E)$ if $\beta =\infty$; explicitly
\begin{equation}
{\zeta}_\omega \ = \ \begin{cases}  F^{(\beta,E_F)}_\omega \ := \
\frac{1}{1+\e^{\beta (H_\omega - E_F)}} \, , &
\beta < \infty \, , \\
P^{(E_F)}_\omega \ := \ \chi_{(-\infty,E_F]}(H_\omega) \, , &\beta = \infty \, .
\end{cases} \label{hyp}
\end{equation}
However we note that our analysis allows for fairly general functions $f$ \cite{BGKS}. We set $\zeta=(\zeta_\omega)_{\omega\in\Omega}\in\K_\infty$ but shall also write $\zeta_\omega$ instead of $\zeta$. That $f$ is the Fermi-Dirac distribution plays no role in the derivation of the linear response. However computing the Hall conductivity itself (once the linear response performed) we shall restrict our attention to the zero temperature case with the \emph{Fermi projection} $P_\omega^{(E_F)}$.

The system is  described by the ergodic time dependent Hamiltonian
$H_\omega(t)$, as in \eqref{eq:Htilde}.   
Assuming the system was in equilibrium at $t=-\infty$ with the density matrix
$\varrho_\omega(-\infty) = {\zeta}_\omega$, the time dependent density matrix
$\varrho_\omega(t)$ is the solution of the Cauchy problem
for the Liouville equation. Since we shall solve the evolution equation in $\L^p(\K_\infty)$, we work with $H(t)=(H_\omega(t))_{\omega\in\Omega}$, as in Assumption~\ref{RandH}.

The electric field $\El_\eta(t)\cdot\x=\e^{\eta t} \El(t)\cdot\x$ is swichted on adiabatically between $t=-\infty$ and $t=t_0$ (typically $t_0=0)$.  Depending on which conductivity on is interested, one may consider different forms for $\El(t)$. In particular $\El(t)=\El$  leads to the direct conductivity, while $\El(t)=cos(\nu t)\El $ leads to the AC-conductivity at frequency $\nu$\footnote{The AC-conductivity may be better defined using the from \eqref{Egen} as argued in \cite{KLM}.}. The first one is  relevant for studying the Quantum Hall effect (see subsection~\ref{subsectHall}), while the second  enters the Mott's formula \cite{KLP,KLM}.

We write
\begin{equation} \label{P(t)}
{\zeta}(t) = G(t) {\zeta} G(t)^* =\G(t) ({\zeta}),
\quad \text{i.e.,}\quad {\zeta}(t)=f(H(t)).
\end{equation}

\begin{theorem}\label{thmrho} Let $\eta>0$ and assume that $\int_{-\infty}^{t} \e^{\eta r} |\El(r)| \mathrm{d}r < \infty$ for all $t\in\R$. Let $p\in[1,\infty[$. Assume that $\zeta\in\W^{1,p}(\K_\infty)$ and that $\nabla\zeta\in\D_{p}^{o}$. The Cauchy problem
\begin{eqnarray}\label{dynamics}
\left\{ \begin{array}{l}i\partial_t \varrho(t) =
[H(t),\varrho(t))]\\
\lim_{t \to  -\infty} \varrho(t)= {\zeta}
\end{array}\right.  \label{dynamicsL}   ,
\end{eqnarray}
 has a unique  solution in $\mathrm{L}^{p}(\K_\infty)$, that is given by
\begin{eqnarray}
\varrho(t)&=&\lim_{s \to -\infty} { \U}(t,s)\left( {\zeta}
\right)
\label{defrho1}\\
&=& \lim_{s \to -\infty}{ \U}(t,s)\left( {\zeta}(s) \right)
\label{defrho2}\\
&=& {\zeta}(t) - 
\int_{-\infty}^t \mathrm{d} r \,\mathrm{e}^{\eta r}{ \U}(t,r)(\mathbf{E}(r)\cdot\nabla{\zeta}(r)). \label{defrho3}
\end{eqnarray}
We  also have
\begin{eqnarray}
\varrho(t) ={ \U} (t,s) (\varrho(s))\, ,
\;\;\|\varrho(t)\|_{p}=\|{\zeta}\|_{p} \, ,
\end{eqnarray}
for all $t,s$. Furthermore, $\varrho(t)$ is
non-negative, and if ${\zeta}$ is a projection, then so is $\varrho(t)$
 for all $t$.
\end{theorem}

\begin{remark}\label{remxHP}
If the initial state ${\zeta}$ is of the form \eqref{hyp}, then the hypotheses of Theorem~\ref{thmrho} hold for any $p>0$, provided ${\zeta}_\omega=P^{(E_F)}_\omega$ that $E_F$ lies in a region of localization.  This is true for suitable $\A_\omega,V_\omega$ and $E_F$, by the methods of, for
example, \cite{CH,Wa, GK1,GK2, GK3,BGK, AENSS,U,GHK} and for the models studied therein as well as in \cite{CH,GK3}. The bound $\E \||\x|\zeta_\omega\chi_0\|^2 <\infty$ or equivalently $\nabla\zeta\in\mathrm{L}^{2}(\K_\infty)$ is actually sufficient for our applications. For $p=1,2$, we refer to \cite[Proposition~4.2]{BGKS} and \cite[Lemma~5.4]{BGKS} for the derivation of these hypotheses from known results.
\end{remark}

\begin{proof}[Proof of Theorem~\ref{thmrho}]
Let us first define
\begin{equation}\label{defrhots}
\varrho(t,s) :={ \U}(t,s)({\zeta}(s)).
\end{equation}
We get, as operators in $\M(\K_\infty)$,
\begin{eqnarray}
\partial_s \varrho(t,s) & =  &
i\U(t,s)\left([H(s),{\zeta}(s)]\right)
+\U(t,s)\left({\El}_\eta(s)\cdot \nabla{\zeta}(s)\right)
\nonumber
\\
& =& \U(t,s)\left( {\El}_\eta(s)\cdot\nabla{\zeta}(s)\right)
\, ,
\end{eqnarray}
where we used \eqref{eqderivtau} and Lemma~\ref{lemGt}.
As a consequence,
with $\mathbf{E}_\eta(r)= \mathrm{e}^{\eta r}\mathbf{E}(r)$,
\begin{equation}
\varrho(t,t) - \varrho(t,s)
  = \int_s^t \mathrm{d} r\,
 \mathrm{e}^{\eta r}{ \U}(t,r)
({\El}(r)\cdot\nabla{\zeta}(r) )
.
\end{equation}
Since $\|{ \U}(t,r) ({\El}(r)\cdot\nabla({\zeta}(r) )\|_p
 \le c_d |{\El}(r)|\, \|\nabla{\zeta}\|_p < \infty$,
the integral is absolutely convergent by hypothesis on $\El_\eta(t)$, and the limit as $s\to-\infty$ can be performed in $ \mathrm{L}^p(\K_\infty)$. It
yields the equality between \eqref{defrho2} and \eqref{defrho3}. Equality of \eqref{defrho1}
and \eqref{defrho2} follows from Lemma~\ref{lemGt} which gives
\begin{equation} {\zeta} = \lim_{s \to -\infty} {\zeta}(s) \;\;
\mbox{in  $ \mathrm{L}^p(\K_\infty)$.}
\end{equation}
 Since $\U(t,s)$ are isometries on $\mathrm{L}^p(\K_\infty)$,
it follows from \eqref{defrho1} that
$\|\varrho(t)\|_p=\|{\zeta}\|_p$. We also get
$\varrho(t) = \varrho(t)^\ast$. Moreover,
\eqref{defrho1} with the limit in $ \mathrm{L}^p(\K_\infty)$  implies that
$\varrho(t)$ is nonnegative.

Furthermore, if ${\zeta}={\zeta}^2$ then $\varrho(t)$ can be seen to be a projection as follows. Note that  convergence in $\mathrm{L}^p$ implies convergence in $\M(\K_\infty)$, so that,
\begin{multline}\varrho(t) = \sideset{}{^{(\tau)}}\lim_{ s \to -\infty}{ \U}(t,s)\left( {\zeta} \right) =\sideset{}{^{(\tau)}}\lim_{ s \to
-\infty}{ \U}(t,s)\left( {\zeta}\right){ \U}(t,s)\left( {\zeta}\right) \\
=\left\{ \sideset{}{^{(\tau)}}\lim_{ s \to -\infty}{ \U}(t,s)\left( {\zeta}\right)\right\}
 \left\{ \sideset{}{^{(\tau)}}\lim_{ s \to -\infty}{ \U}(t,s)\left( {\zeta}\right)\right\}\ = \
 \varrho(t)^2\, .
\end{multline}
 where we note $\ \sideset{}{^{(\tau)}}\lim$ the limit in the topological algebra $\M(\K_\infty)$.

 To see that $\varrho(t)$ is a solution of  \eqref{dynamics}
 in $\mathrm{L}^p(\K_\infty)$, we differentiate the expression \eqref{defrho3} using
 \eqref{eqderivU3} and Lemma~\ref{lemGt}.
We get
\begin{eqnarray}
i\partial_t  \varrho(t)& =& -
 \int_{-\infty}^t \mathrm{d} r \,\mathrm{e}^{\eta r}
\left[H(t),{ \U}(t,r) \left( \mathbf{E}(r)
\cdot \nabla {\zeta}(r) \right)\right] \label{partiavarho} \\
& =&-\left[H(t), \left\{
 \int_{-\infty}^t \mathrm{d} r \,\mathrm{e}^{\eta r}{ \U}(t,r) \left( \mathbf{E}(r)
\cdot \nabla {\zeta}(r) \right)\right\}\right] \label{partiavarho2}\\
& =&\left[H(t),  \left\{{\zeta}(t) - 
 \int_{-\infty}^t \mathrm{d} r \,\mathrm{e}^{\eta r}{ \U}(t,r) \left( \mathbf{E}(r)
\cdot \nabla {\zeta}(r) \right)\right\}\right] \nonumber\\
& =&\left[H(t), \varrho(t)\right] \label{partiavarho3} .
\end{eqnarray}
The  integral in \eqref{partiavarho} converges since by \eqref{HUAWbound} ,
\begin{equation}\label{intderivative}
\|\left[H(t),{ \U}(t,r) \left( \mathbf{E}(r)
\cdot \nabla {\zeta}(r) \right)\right]\|_p \le 2 C \| (H +\gamma) (\mathbf{E}(r) \cdot \nabla {\zeta}) \|_p    .
\end{equation}
Then we justify going from
\eqref{partiavarho} to \eqref{partiavarho2} by inserting a resolvent $(H(t) +\gamma)^{-1}$ and making use of \eqref{HUAW}.

It remains to show that the solution of \eqref{dynamicsL} is unique in
$\mathrm{L}^p(\K_\infty)$. It suffices to show that if
$\nu(t)$ is a solution of \eqref{dynamicsL} with $\zeta =0$ then
$\nu(t)= 0$ for all $t$. We define $\tilde{\nu}^{(s)}(t)  = \U(s,t)(\nu(t))$ and proceed by duality. Since $p\ge 1$, with pick $q$ s.t. $p^{-1}+q^{-1}=1$.
 If $A \in \D_q^{(0)}$, we have, using Lemma~\ref{lemcyclicity},
\begin{align}
&i\partial_t \T \left\{ A \tilde{\nu}^{(s)}(t)\right\}=
i\partial_t \T \left\{{ \U}(t,s)(A) {\nu}(t)\right\}\\
& \quad  = \T \left\{[H(t),\U(t,s)(A)] \
{\nu}(t)\right\} + \T \left\{{ \U}(t,s)(A)
\Ll_{q}(t)(\nu(t))\right\} \nonumber
 \\
 & \quad = -\T \left\{\U(t,s)(A) 
\Ll_{q}(t)({\nu}(t))\right\} + \T \left\{{ \U}(t,s)(A)\
\Ll_{q}(t)(\nu(t))\right\} = 0\, . \nonumber
\end{align}
We conclude that for all $t$ and $A \in \D^{(0)}_q$ we have
  \begin{equation}\label{Tanu}
\T \left\{ A \tilde{\nu}^{(s)}(t)\right\}=
\T \left\{ A \tilde{\nu}^{(s)}(s)\right\}=
\T \left\{ A {\nu}(s)\right\}.
\end{equation}
Thus  $\tilde{\nu}^{(s)}(t)  = \nu(s)$ by Lemma~\ref{lemduality},
that is,  $\nu(t)=\U(t,s)(\nu (s))$.
 Since by hypothesis $\lim _{s \rightarrow -\infty} \nu(s) =0$, we obtain
that $\nu(t)=0$ for all $t$.
\end{proof}



\subsection{The current and the conductivity}

The velocity operator $\mathbf{v}$ is defined as 
\begin{equation}\label{velocity}
\v = \v(\A) = 2 \Db(\A)  ,
\end{equation}
where $\Db=\Db(\A)$ is defined below \eqref{HtH}. Recall that $\v=2(-i\nabla - \A)=i[H,\x]$ on $\Cc^\infty(\R^d)$. We also set $\Db(t)= \Db(\A+\F_\eta(t))$ as in \eqref{defDAt}, and $\v(t)=2\Db(t)$.

From now on $\varrho(t)$ will denote the unique solution to
\eqref{dynamicsL}, given explicitly in \eqref{defrho3}. 
If $H(t) \varrho(t) \in \mathrm{L}^p(\K_\infty)$ then clearly
$\Db_{j}(t)\varrho(t)$ can be defined as well by
\begin{eqnarray}\label{Dvarrho}
\Db_{j}(t)\varrho(t) = \left(\Db_{j}(t)
 (H(t)+\gamma)^{-1}\right) \left( (H(t)+\gamma)
\varrho(t)\right) ,
\end{eqnarray}
since $\Db_{j}(t)  (H(t)+\gamma)^{-1}\in \K_\infty$, and thus $\Db_{j}(t)\varrho(t)\in \mathrm{L}^p(\K_\infty)$.

\begin{definition}\label{defcurrent}
Starting with a system in equilibrium in state ${\zeta}$,
the net current (per
unit volume), $\J(\eta,\El;{\zeta},t_0)\in\R^d$, generated by switching on an electric
field $\El$ adiabatically  at rate $\eta>0$  between time $-\infty$ and time
$t_0$,  is defined as
\begin{equation}\label{defJ}
\J(\eta,\El;{\zeta},t_0) =  \T \left( \v(t_0) \varrho(t_0)\right)
 - \T \left( \v {\zeta}\right)  .
\end{equation}
\end{definition}

As it is well known, the current is null at equilibrium:

\begin{lemma}\label{lemequilib}
One has $\T(\Db_{j} \zeta)=0$ for all $j=1,\cdots, d$, and thus $  \T \left( \v \zeta \right)=0$.
\end{lemma}

Throughout the rest of this paper, we shall assume that the electric field has the form \begin{equation}\label{hypE}
\El(t)=\mathcal{E}(t)\El,
\end{equation} 
where $\El\in\C^d$
 gives the intensity of the electric in each direction while $|\mathcal{E}(t)|=\mathcal{O}(1)$ modulates this intensity as time varies. As pointed out above, the two cases of particular interest are $\mathcal{E}(t)=1$ and $\mathcal{E}(t)= \cos(\nu t)$. We may however, as in \cite{KLM}, use the more general form
\begin{equation}\label{Egen}
\mathcal{E}(t) = \int _\mathbb{R}  \cos(\nu t) \hat{\mathcal{E}}(\nu) \mathrm{d}\nu,
\end{equation} 
for suitable $\hat{\mathcal{E}}(\nu)$ (see \cite{KLM}).

 It is useful to rewrite the current \eqref{defJ}, using \eqref{defrho3} and Lemma~\ref{lemequilib}, as
\begin{eqnarray} 
\J(\eta,\El;{\zeta},t_0)&  = & \T \left\{ 2\Db(0)
\left( \varrho(t_0) - {\zeta}(t_0)\right)\right\} \label{formuleJ} \\
 & = &-\T \left\{2
 \int_{-\infty}^{t_0}  \mathrm{d} r\,  \mathrm{e}^{\eta r}\Db(0) \,{ \U}(t_0,r) \left(
 \mathbf{E}(r) \cdot \nabla {\zeta}(r)  \right) \right\}  . \nonumber \\
 & = &-\T \left\{2
 \int_{-\infty}^{t_0}  \mathrm{d} r\,  \mathrm{e}^{\eta r} \mathcal{E}(r) \Db(0) 
 \,{ \U}(t_0,r) \left(
 \mathbf{E} \cdot \nabla {\zeta}(r)  \right) \right\}  . \notag
\end{eqnarray}

The conductivity tensor $\sigma(\eta;{\zeta},t_0)$ is  defined as the derivative  of  the function $\J(\eta,\El;{\zeta},t_0)\colon\R^d \to \R^d$ at
$\El=0$.  Note that
 $\sigma(\eta;{\zeta},t_0)$ is a $d \times d$ matrix
 $\left\{\sigma_{jk}(\eta;{\zeta},t_0)\right\}$.
\begin{definition}\label{defsigma}
For $\eta >0$ and $t_0\in \R$,  the conductivity tensor $\sigma(\eta;{\zeta},t_0)$ is defined as
\begin{equation}\label{defsigmajketa}
\sigma(\eta;{\zeta},t_0) = \partial_{\El} (\J(\eta,\El;{\zeta},t_0))_{\mid\El =0} \, ,
\end{equation}
if it exists. The conductivity tensor $\sigma({\zeta},t_0)$
is
defined by
\begin{equation}\label{defsigmajk}
\sigma(\zeta,t_0) :=  \lim_{\eta\downarrow 0} \sigma(\eta;\zeta,t_0)\, ,
\end{equation}
whenever the limit exists.
\end{definition}

\subsection{Computing the linear response:  a Kubo formula for the conductivity}

The next theorem gives  a ``Kubo formula" for the conductivity at positive adiabatic parameter.
\begin{theorem}\label{thmsgmjk}  Let $\eta > 0$. Under the hypotheses of Theorem~\ref{thmrho} for $p=1$, the current $\J(\eta,\El;\zeta,t_0)$ is differentiable with
respect to $\El$ at $\El=0$ and the derivative $\sigma(\eta;{{\zeta}})$ is given by
\begin{eqnarray}\label{sigmajk}
\sigma_{jk}(\eta;\zeta,t_0) 
&=& \, -\T \left\{ 2  \int_{-\infty}^{t_0}  \mathrm{d} r\,
\mathrm{e}^{\eta r} \mathcal{E}(r)
  \Db_{j} \, \U^{(0)}(t_0-r) \left(\partial_k( {\zeta}) \right)  \right\} .
\end{eqnarray}
\end{theorem}

The analogue of \cite[Eq.~(41)]{BESB} and
\cite[Theorem~1]{SBB2} then holds:
 
\begin{corollary}\label{corsgmjk} Assume that $\mathcal{E}(t)=\Re\e^{i\nu t}$, $\nu\in\R$, then the conductivity  $\sigma_{jk}(\eta;{{\zeta}};\nu) $ at frequency $\nu$ is given by
\begin{eqnarray}\label{sigmajkbis}
\sigma_{jk}(\eta;\zeta;\nu;0) 
&=&\, - \T \left\{ 2 \Db_{j} \,
(i\mathcal{L}_1 +\eta +i\nu)^{-1} \left( \partial_k \zeta )\right\}\right\rbrace\ ,
\end{eqnarray}
\end{corollary}
\begin{proof}[Proof of corollary~\ref{corsgmjk}]
 Recall (4.11) , in particular $\zeta=\zeta^{\frac{1}{2}}\zeta^{\frac{1}{2}}$. It follows that $\sigma(\eta;\nu;\zeta;0)$ in (4.27) is real (for arbitrary $\zeta=f(H)$ write $f=f_+-f_-$). As a consequence ,
 \begin{eqnarray}
\sigma(\eta;\nu;\zeta;0)=-\Re\T \left\{ 2  \int_{-\infty}^{t_0}  \mathrm{d} r\,
\mathrm{e}^{\eta r} \mathrm{e}^{i\nu r}
\Db_{j} \, \U^{(0)}(t_0-r) \left(\partial_k( {\zeta}) \right)  \right\} .
\end{eqnarray}
Integrating over $r$ yields the result.
\end{proof}
  
\begin{proof}[Proof of Theorem~\ref{thmsgmjk}]
For clarity, in this proof we display the argument $\El$ in all
functions which depend on $\El$. From \eqref{formuleJ} and $\J_j(\eta,0;\zeta,t_0)= 0$ (Lemma \ref{lemequilib}), we
have
\begin{equation}\label{defsigmajk3}
    \sigma_{jk}(\eta;\zeta,t_0)=-\lim_{E
\to 0} 2\T \left\{
 \int_{-\infty}^{t_0}  \mathrm{d} r\,  \mathrm{e}^{\eta r}\mathcal{E}(r)
\Db_{\El,j}(0)
 \U(\El,0,r) \left( \partial_k \zeta(\El,r) \right) \right\} .
\end{equation}

First understand we can interchange integration and the limit $\El
\rightarrow 0$, and get
\begin{equation}
\label{defsigmajk31}
    \sigma_{jk}(\eta;\zeta,t_0) \ = \ -2 \int_{-\infty}^{t_0}  \mathrm{d} r\,  \mathrm{e}^{\eta r}\mathcal{E}(r)
 \lim_{{E} \to 0}  \T \left\{  \Db_{j}({\El}, 0)
{ \U}({\El},0,r) \left(\partial_k  {\zeta}({\El}, r)  \right)
\right\}\, .
\end{equation}
The latter can  easily be seen by inserting a resolvent $(H(t) +\gamma)^{-1}$ and making use of \eqref{HUAW}, the fact that $H\nabla\zeta\in \mathrm{L}^1(\K_\infty)$, the inequality : $\vert\T(A)\vert\leq\T(\vert A\vert)$
and dominated convergence.

Next, we note that for any $s$ we have
\begin{eqnarray}\label{GEk}
\lim_{{E} \to 0}\G({\El}, s)= I \;\;\mbox{strongly in $\mathrm{L}^1(\K_\infty)$}\, ,
\end{eqnarray}
which can be proven  by a argument similar to the one used to prove
Lemma~\ref{lemGt}.
 Along the same lines, for $B \in \K_\infty$   we have
\begin{eqnarray}\label{GEk2}
\lim_{{E} \to 0}\G({\El}, s)(B_\omega)= B_\omega  \;\;\mbox{strongly in $\H$,
with  $\|\G({\El}, s)(B)\|_\infty=
\|B\|_\infty $}\, .
\end{eqnarray}
Recalling that $\Db_{j,\omega}({\El}, 0) = \Db_{j,\omega} - \F_j(0)$ and that $\|\partial_k  {\zeta}({\El}, r)\|_1=\|\partial_k {\zeta}\|_1<\infty$, using Lemma~\ref{lemlimstrong},
\begin{align}
\lim_{{E} \to 0} & \T \left\{  \Db_{j}({\El}, 0){ \U}({\El},0,r) \left(\partial_k {\zeta}({\El}, r)
  \right)\right\}  = \lim_{{\El} \to 0}  \T \left\{
    \Db_{j}  
    U({\El}, 0,r) (\partial_k {\zeta})   U({\El},r,0)
    \right\} \nonumber \\
 & = \lim_{{\El} \to 0}  \T \left\{
    \Db_{j} U({\El}, 0,r) (\partial_k {\zeta} ) 
    U^{(0)}(r) \right\} \label{limit101},
\end{align}
where we have inserted (and removed) the resolvents $(H(\El,r) + \gamma)^{-1}$ and
$(H +\gamma)^{-1}$.

To proceed it is convenient to introduce a cutoff so that we can deal with
$\Db_{j}$ as if it were in $\K_\infty$.  Thus we pick $f_n\in
C^\infty_c(\R)$, real valued, $|f_n|\le 1$, $f_n=1$ on $[-n,n]$, so that $f_n(H)$ converges strongly to $1$. Using Lemma~\ref{lemlimstrong}, we have
\begin{align}
& \T \left\{
    \Db_{j} U({\El}, 0,r) (\partial_k {\zeta} ) 
    U^{(0)}(r) \right\} 
     =   \lim_{n\to\infty} \T \left\{f_n(H) \Db_{j}
    U({\El},0,r)  (\partial_k {\zeta} )  U^{(0)}(r)
    \right \} \nonumber \\
& =  \lim_{n\to\infty} \T \left\{  U({\El},0,r) 
    \left( (\partial_k {\zeta} )(H+\gamma) \right)
U^{(0)}(r) ( H+\gamma)^{-1} f_n(H) \Db_{j} \right \}  \nonumber \\
& = \label{DUQ6} \T \left\{  U({\El},0,r)  \left(
(\partial_k {\zeta})(H+\gamma)  \right)
\left( U^{(0)}(r) (H+\gamma )^{-1} \Db_{j} \right)\right \} \; ,
\end{align}
where we used  the centrality of the trace, the fact that
$ (H+\gamma )^{-1}$ commutes with $U^{(0)}$  and then that
 $ ( H+\gamma)^{-1}  \Db_{j} \in \K_\infty$ in order to remove to limit $n\to\infty$.
  Finally, combining \eqref{limit101} and \eqref{DUQ6}, we get
\begin{align}
\lim_{{E} \to 0} & \T \left\{  \Db_{j}({\El}, 0){ \U}({\El},0,r) \left(\partial_{k}{\zeta}({\El}, r)  \right)\right\}
 \\ & = \ \T \left\{  U^{(0)}(-r) \ \left(
( \partial_{k} {\zeta})(H+\gamma) \right)
U^{(0)}(r) (H+\gamma)^{-1}  \Db_{j} 
\right \} \nonumber \\
& = \T \left\{  \Db_{j}{ \U}^{(0)}(-r) (\partial_{k} {\zeta}) \right \}\label{finally} .
\end{align}
The Kubo formula \eqref{sigmajk} now follows from
\eqref{defsigmajk31}  and \eqref{finally}.
\end{proof}



\subsection{The  Kubo-St\u{r}eda formula for the Hall conductivity}
\label{subsectHall}

Following \cite{BESB,AG}, we now recover the well-known  Kubo-St\u{r}eda formula
 for the Hall
conductivity at zero temperature (see however Remark~\ref{remAC} for AC-conductivity). To that aim we consider  the case $\mathcal{E}(t)=1$ and $t_0=0$. Recall Definition~\ref{defsigma}. We write
\begin{equation}
    \sigma_{j,k}^{(E_f)} = \sigma_{j,k}(P^{(E_F)},0) \; , \text{ and } \
    \sigma_{j,k}^{(E_f)}(\eta) = \sigma_{j,k}(\eta;P^{(E_F)},0) \; .
\end{equation}

\begin{theorem}\label{thmHall}
Take $\mathcal{E}(t)=1$ and $t_0=0$. If ${\zeta}= P^{(E_F)}$ is a
Fermi projection satisfying the hypotheses of Theorem~\ref{thmrho} with $p=2$, we have
\begin{eqnarray}\label{expHall}
\sigma_{j,k}^{(E_F)}
=-i \T \left\{ P^{(E_F)}\left[  \partial_j P^{(E_F)} ,
 \partial_k P^{(E_F)}\right
] \right\},
\end{eqnarray}
for all $j,k=1,2,\ldots,d$.  As a consequence,  the conductivity tensor is antisymmetric;
in particular $\sigma_{j,j}^{(E_F)} =0$ for $j=1,2,\ldots,d$.
\end{theorem}

Clearly the direct conductivity vanishes, $\sigma_{jj}^{(E_F)}=0$. Note that,
if the system is time-reversible the off diagonal elements are zero in the
region of localization, as expected.
\begin{corollary} \label{corHall}
Under the assumptions of Theorem \ref{thmHall}, if $\A=0$ (no magnetic field),
we have $\sigma_{j,k}^{(E_F)} =0$ for all $j,k=1,2,\ldots,d$.
\end{corollary}

We have the crucial following lemma for computing the Kubo-St\u{r}eda formula, which already appears in \cite{BESB} (and then in \cite{AG}).

\begin{lemma}\label{lemtriplecom} 
Let $P\in \K^\infty$ be a projection such that $\partial_k P\in \L^p(\K^\infty)$, then as operators in $\M(\K^\infty)$ (and thus in $\mathrm{L}^p(\K^\infty)$),
\begin{equation}\label{triplecom}
\partial_k P= \left[P, [P, \partial_k P] \right] .
\end{equation}
\end{lemma}

\begin{proof}
Note that $\partial_k P  =  \partial_k P^2 = P\partial_k P + (\partial_k P) P$ so that multiplying left and right both sides by $P$ implies that $P(\partial_k P) P=0$. We then have, in $\L^p(\K^\infty)$,
\begin{eqnarray*}
\partial_k P & = &  P\partial_k P + (\partial_k P) P =  P\partial_k P + (\partial_k P) P - 2P(\partial_k P) P\\
& = & P (\partial_k P) (1-P) + (1-P)(\partial_k P) P \\
& = & \left[P, [P, \partial_k P] \right] .
\end{eqnarray*}
\end{proof}

Remark that
Lemma~\eqref{lemtriplecom} heavily relies on the fact $P$ is a projection. We shall apply it to the situation of zero temperature, i.e. when
the initial density matrix is the  orthogonal projection $ P^{(E_F)}$. The argument
would not go through at positive temperature.

\begin{proof}[Proof of Theorem~\ref{thmHall}] We again regularize the velocity
$\Db_{j,\omega}$ with a smooth function $f_n\in\mathcal{C}^\infty_c(\R)$,
$|f_n|\le 1$, $f_n=1$ on $[-n,n]$, but this time we also require that $f_n=0$ outside $[-n-1,n+1]$, so that $f_n \chi_{[-n-1,n+1]}=f_n$. Thus
 $\Db_{j} f_n(H)\in  \mathrm{L}^{p,o}(\K_\infty)$,  $0<p\le \infty$. Moreover
\begin{equation}\label{propfn}
f_n(H) (2\Db_{j}) f_n(H) = f_n(H) P_n (2\Db_{j}) P_n f_n(H) = - f_n(H) \partial_j (P_n H) f_n(H)
\end{equation}
where $P_n=P_n^2=\chi_{[-n-1,n+1]}(H)$ so that $H P_n$ is a bounded operator.
We have, using the centrality of the trace
 $\T$ , that
\begin{eqnarray}\label{eqhall00a}
\widetilde \sigma_{jk}^{(E_F)}(r)& :=  & -\T\left\{ 2\Db_{j,\omega}
\U^{(0)}(-r) ( \partial_{k} P^{(E_F)})  \right\}
\\
& = &   \label{eqhall00} -\lim_{n\to\infty} \T\left\{
\U^{(0)}(r)(f_n(H)2\Db_{j,\omega}f_n(H))   \partial_{k} P^{(E_F)} \right\} .
\end{eqnarray} 
Using Lemma~\ref{formulaCommutator} and applying Lemma~\ref{lemtriplecom} applied to $P=P^{(E_F)}$, it follows that
\begin{eqnarray}\label{eqhall009}
\lefteqn{\T\left\{ {\U^{(0)}}(r)(f_n(H)2\Db_{j}f_n(H))   \partial_{k} P^{(E_F)}  \right\}
 }\\ & = &
\T\left\{ {\U^{(0)}}(r)(f_n(H)2\Db_{j}f_n(H)) 
\left[P^{(E_F)}, \left[P^{(E_F)}, \partial_{k} P^{(E_F)}
\right]\right]  \right\} \nonumber
\\ & = &
\T\left\{{\U^{(0)}}(r)\left( \left[P^{(E_F)},
\left[P^{(E_F)}, f_n(H)2\Db_{j}f_n(H)
\right]\right] \right)     \partial_{k} P^{(E_F)} \right\} ,\nonumber
\\ & = &
- \T\left\{{\U^{(0)}}(r)\left( \left[P^{(E_F)},
f_n(H)\left[P^{(E_F)}, \partial_j (HP_n) 
\right]f_n(H)\right] \right)     \partial_{k} P^{(E_F)} \right\} ,\nonumber
\end{eqnarray}
where we used that $P^{(E_F)}$ commutes with $U^{(0)}$ and $f_n(H)$, and \eqref{propfn}.
Now, as elements in $\M(\K^\infty)$,
\begin{equation}\label{equalcom}
\left[ P^{(E_F)}, \partial_j HP_n \right]
=\left[H P_n , \partial_j P^{(E_F)} \right]  .
\end{equation}
Since $[H , \partial_j P^{(E_F)}]]$ is well defined by hypothesis, $f_n(H)\left[H P_n , \partial_j P^{(E_F)}\right]f_n(H)$ converges in $L^p$ to the latter as $n$ goes to infinity. Combining \eqref{eqhall00}, \eqref{eqhall009},
and  \eqref{equalcom}, we get after taking $n \rightarrow \infty$,
\begin{equation}\label{eqhall0090} \widetilde \sigma_{jk}^{(E_F)}(r) \
 = \
-\T\left\{{\U^{(0)}}(r)\left( \left[P^{(E_F)}, \left[H ,
\partial_j P^{(E_F)} \right] \right] \right)  \partial_k
 P^{(E_F)}\right\} \, .
\end{equation}
Next,
\begin{eqnarray}
\left[ P^{(E_F)},\left[H , \partial_j P^{(E_F)}
\right] \right]  =  \left[H, \left[ P^{(E_F)}
, \partial_j P^{(E_F)} \right] \right] ,
\end{eqnarray}
so that, recalling Proposition~\ref{Liouvillian},
\begin{eqnarray}\nonumber
\widetilde \sigma_{jk}^{(E_F)}(r)& = & -\T\left\{\U^{(0)}(r)\left(
\left[H, \left[ P^{(E_F)} ,
 \partial_j P^{(E_F)} \right]
\right] \right)  \partial_k P^{(E_F)}\right\}\\
&=& 
-\left \la  \e^{-ir\Ll} { \Ll_2} \left ( \left[
P^{(E_F)} , \partial_j P^{(E_F)} \right] \right ),
        \partial_k P^{(E_F)} \right\ra_{\L^2}  \label{eqhall01} ,
\end{eqnarray}
where $\la A,B\ra_{\L^2}=\T(A^\ast B)$. Combining \eqref{sigmajk}, \eqref{eqhall00a}, and \eqref{eqhall01},
we get
\begin{equation} \label{sigmaL}
\sigma_{jk}^{(E_F)}(\eta)  = - \left \la    i  \left(\Ll_2+
i\eta \right)^{-1}{ \Ll_2} \left ( \left[ P^{(E_F)} ,
  \partial_j P^{(E_F)} \right] \right ),
         \partial_k P^{(E_F)} \right\ra_{\L^2}  .
\end{equation}
It follows from the spectral theorem  (applied to
$\mathcal{L}_2$) that
\begin{equation}\label{stL}
\lim_{\eta\to 0} \left(\Ll_2+ i\eta \right)^{-1}{ \Ll_2} =
P_{(\mathrm{Ker}\, \mathcal{L}_2)^\perp}  \;\;\mbox{strongly in $\L^2(\K_\infty)$} \, ,
\end{equation}
where
 $P_{(\mathrm{Ker}\, \mathcal{L}_2)^\perp}$
is the orthogonal projection onto ${(\mathrm{Ker} \,\mathcal{L}_2)^\perp}$.  Moreover, as in \cite{BGKS} one can prove that
\begin{equation}\label{Pperp}
\left[ P^{(E_F)} ,  \partial_j P^{(E_F)} \right]
\in{(\mathrm{Ker}\, \mathcal{L}_2)^\perp}\, .
 \end{equation}
Combining \eqref{sigmaL}, \eqref{stL},  \eqref{Pperp}, and
 Lemma~\ref{formulaCommutator}, we get
\begin{equation}\nonumber
 \sigma_{j,k}^{(E_F)} = i  \left \la 
\left[ P^{(E_F)} , \partial_j P^{(E_F)} \right],
        \partial_k P^{(E_F)}  \right\ra_{\L^2} 
=  -i \T\left\{    P^{(E_F)}  \left[ \partial_j P^{(E_F)} ,
       \partial_k P^{(E_F)} \right]\right\}\,, \nonumber
\end{equation}
which is just  \eqref{expHall}.  
\end{proof}

\begin{remark}\label{remAC}
If one is interested in the AC-conductivity, then the proof above is valid up to \eqref{sigmaL}. In particular, with $\mathcal{E}(t)=\Re\e^{i\nu t}$, one obtains
\begin{equation} \label{sigmaLac}
\sigma_{jk}^{(E_F)}(\eta)  = - \Re\left \la    i  \left(\Ll_2+\nu+
i\eta \right)^{-1}{ \Ll_2} \left ( \left[ P^{(E_F)} ,
  \partial_j P^{(E_F)} \right] \right ),
         \partial_k P^{(E_F)} \right\ra_{\L^2}  .
\end{equation}
The limit $\eta\to 0$ can still be performed as in \cite[Corollary~3.4]{KLM}. It is the main achievement of  \cite{KLM} to be able to investigate the behaviour of this limit as $\nu\to 0$ in connection with Mott's formula.
\end{remark}




\begin{thebibliography}{ThKNN}


\bibitem[AENSS]{AENSS} Aizenman, M., Elgart, A., Naboko, S., Schenker, J.H.,
Stolz, G.: Moment Analysis for Localization in Random Schr\"odinger Operators.
2003 Preprint, math-ph/0308023.


\bibitem[AG]{AG} {Aizenman, M.,   Graf, G.M.:}
{ Localization bounds for an electron gas}, J. Phys. A: Math.
Gen. {\bf 31}, 6783-6806, (1998).


\bibitem[AvSS]{ASS} Avron, J., Seiler, R., Simon, B.: Charge deficiency,
charge transport and comparison of dimensions.  Comm. Math. Phys.~{\bf 159},
399-422 (1994).

\bibitem[B]{Be}  Bellissard, J.: Ordinary quantum {H}all effect and
noncommutative cohomology. In {\em Localization in disordered systems (Bad
Schandau, 1986)}, pp. 61-74. Teubner-Texte Phys. {\bf 16}, Teubner, 1988.

\bibitem[BES]{BESB} Bellissard, J., van Elst, A., Schulz-Baldes, H.:
The non commutative geometry of the quantum Hall effect.
{J. Math. Phys.}~{\bf 35}, 5373-5451 (1994).

\bibitem[BH]{BH} Bellissard, J., Hislop, P.: Smoothness of correlations in the Anderson model at strong disorder.  Ann. Henri Poincar�  8, 1-26 (2007).



\bibitem[BoGK]{BGK}  Bouclet, J.M.,   Germinet, F., Klein, A.:
Sub-exponential decay of operator kernels
 for functions of generalized Schr\"odinger operators.
Proc. Amer. Math. Soc.    \textbf{132} ,  2703-2712  (2004).

\bibitem[BoGKS]{BGKS}  Bouclet, J.M.,   Germinet, F., Klein, A., Schenker. J.:
Linear response theory for magnetic Schr�dinger operators in disordered media, J. Funct. Anal. \textbf{226}, 301-372 (2005)






\bibitem[CH]{CH} Combes,  J.M.,   Hislop, P.D.:{ Landau Hamiltonians with
random potentials: localization and the density of states}. Commun. Math.
Phys. {\bf 177}, 603-629 (1996).

\bibitem[CGH]{CGH} Combes, J.M.,   Germinet, F., Hislop, P.D.: Conductivity and current-current correlation measure. In preparation.



\bibitem[CoJM]{CJM} Cornean, H.D., Jensen, A., Moldoveanu, V.:
The Landauer-B\"uttiker formula and resonant quantum transport.  Mathematical physics of quantum mechanics,  45-53, Lecture Notes in Phys., 690, Springer, Berlin, 2006.

\bibitem[CoNP]{CNP} Cornean, H.D., Nenciu, G. Pedersen, T.: The Faraday effect revisited: general theory.  J. Math. Phys.  47  (2006),  no. 1, 013511, 23 pp.






\bibitem[D]{Dix} Dixmier, J.: Les alg\`ebres d'op\'erateurs dans l'espace Hilbertien (alg\`ebres de von Neumann), Gauthier-Villars 1969 and Gabay 1996.

\bibitem[Do]{Do} N. Dombrowski. PhD Thesis. In preparation.




\bibitem[ES]{ES} Elgart, A. , Schlein, B.: Adiabatic charge transport
 and the Kubo formula for Landau Type Hamiltonians.
 Comm. Pure Appl. Math.  57,  590-615 (2004).



\bibitem[FS]{FS}  Fr\"ohlich, J.,  Spencer, T.: {Absence of diffusion with
Anderson tight binding model for large disorder or low energy}. Commun.
Math. Phys. {\bf 88}, 151-184 (1983)


\bibitem[Geo]{Geo} Georgescu, V.: Private communication.

\bibitem[GK1]{GK1}  Germinet, F., Klein, A.: {Bootstrap Multiscale Analysis
and Localization in Random Media}. Commun. Math. Phys.~{\bf  222}, 415-448
(2001).

\bibitem[GK2]{GK2}  Germinet, F.,  Klein, A.: Operator kernel estimates for  functions of generalized Schr\"odinger operators.
Proc. Amer. Math. Soc. 131,  911-920  (2003).

\bibitem[GK3]{GK3} Germinet, F,  Klein, A.:
Explicit finite volume criteria for localization in continuous
 random media and applications. Geom. Funct. Anal. 13, 1201-1238 (2003).


\bibitem[GrHK]{GHK} Ghribi, F., Hislop, P., Klopp, F.: Localization for Schr\"odinger operators with random vector potentials. Contemp. Math. To appear.

\bibitem[KLP]{KLP} Kirsch, W. Lenoble, O. Pastur, L.: On the Mott formula for the ac conductivity and binary correlators in the strong localization regime of disordered systems.  J. Phys. A  36, 12157-12180 (2003).






\bibitem[KLM]{KLM} Klein, A., Lenoble, O., M\"uller, P.: On Mott's formula for the AC-conductivity in the Anderson model. Annals of Math. To appear.

\bibitem[KM]{KM} Klein, A, M\"uller, P.: The conductivity measure for the Anderson model. In preparation.

\bibitem[Ku]{Ku} Kunz, H.:
{The Quantum Hall Effect for Electrons in a Random Potential}.
Commun. Math. Phys.~{\bf 112}, 121-145 (1987).

\bibitem[LS]{LS} { H. Leinfelder, C.G.  Simader},
{\it Schr\"odinger operators with singular magnetic potentials},
Math. Z. 176, 1-19 (1981).

\bibitem[MD]{MD} Mott, N.F., Davies, E.A.: Electronic processes in non-crystalinne materials. Oxford: Clarendon Press 1971.

\bibitem[NB]{NB} Nakamura, S., Bellissard, J.:
{Low Energy Bands do not Contribute to Quantum Hall Effect}.
Commun. Math. Phys.~{\bf 131}, 283-305 (1990).

\bibitem[Na]{Na} Nakano, F.: Absence of transport in Anderson localization.
 Rev. Math. Phys.~{\bf 14},  375-407 (2002).
 
\bibitem[P]{Pa}  Pastur, L., Spectral properties of disordered systems in the one-body approximation.  Comm. Math. Phys.~{\bf  75}, 179-196  (1980).

\bibitem[PF]{PF}  Pastur, L.,  Figotin, A.: {\em Spectra of Random and
Almost-Periodic Operators}.   Springer-Verlag, 1992.






\bibitem[SB1]{SBB1} Schulz-Baldes, H., Bellissard, J.:
Anomalous transport: a mathematical framework.
Rev. Math. Phys.~{\bf 10}, 1-46 (1998).

\bibitem[SB2]{SBB2} Schulz-Baldes, H., Bellissard, J.: A Kinetic Theory
 for Quantum Transport in Aperiodic Media. J. Statist. Phys.~{\bf 91},
 991-1026 (1998).




\bibitem[St]{St} St\u{r}eda, P.: {Theory of quantised Hall conductivity
    in two dimensions}. J. Phys. C. {\bf 15}, L717-L721 (1982).



\bibitem[Te]{Te} Terp, M.: $\mathrm{L}^p$ spaces associated with von Neumann algebras. Notes, Math. Instititue, Copenhagen university 1981.

\bibitem[ThKNN]{TKNN}  Thouless, D. J., Kohmoto, K., Nightingale, M. P.,
den Nijs, M.: {Quantized Hall conductance in a two-dimensional periodic
potential}. Phys. Rev. Lett. {\bf 49}, 405-408 (1982).

\bibitem[U]{U} Ueki, N.: Wegner estimates and localization for Gaussian random potentials.  Publ. Res. Inst. Math. Sci.  40, 29-90 (2004).

\bibitem[W]{Wa} Wang, W.-M.: {Microlocalization, percolation, and
Anderson localization for the magnetic Schr\"odinger operator with a
random potential}.  J. Funct. Anal.~\textbf{146}, 1-26  (1997).

\bibitem[Y]{Y} Yosida, K.:  \emph{Functional Analysis, 6th edition.}
Springer-Verlag, 1980.



\end{thebibliography}
\end{document}